\documentclass[10pt]{amsart}
\usepackage{amsfonts}
\usepackage{amsmath}
\usepackage{amssymb}
\usepackage{graphicx}%
\usepackage{dcolumn}
\usepackage{bm}

\newcommand{\beq}{\begin{equation}}
\newcommand{\eeq}{\end{equation}}
\newcommand{\bea}{\begin{eqnarray}}
\newcommand{\eea}{\end{eqnarray}}
\newcommand{\beqa}{\begin{eqnarray}}
\newcommand{\eeqa}{\end{eqnarray}}
\newcommand{\nn}{\nonumber}
\newcommand\noi{\noindent}

\newtheorem{theorem}{Theorem}

\newtheorem{conjecture}[theorem]{Conjecture}
\newtheorem{corollary}[theorem]{Corollary}

\newtheorem{definition}[theorem]{Definition}

\newtheorem{lemma}[theorem]{Lemma}

\newtheorem{proposition}[theorem]{Proposition}
\newtheorem{remark}[theorem]{Remark}

\begin{document}

\title{Galois Differential Algebras and Categorical Discretization of Dynamical Systems}
\author{Piergiulio Tempesta}
\address{Departamento de F\'{\i}sica Te\'{o}rica II (M\'{e}todos Matem\'{a}ticos de la f\'isica), Facultad de F\'{\i}sicas, Universidad
Complutense de Madrid, 28040 -- Madrid, Spain and Instituto de Ciencias Matem\'aticas, C/ Nicol\'as Cabrera, No 13--15, 28049 Madrid, Spain}
\email{p.tempesta@fis.ucm.es, piergiulio.tempesta@icmat.es}

\begin{abstract}
The problem of integrability-preserving discretization of dynamical systems with variable coefficients is solved by means of an approach based on category theory. In this setting, a differential equation and its discrete analogs are considered to be different representations of the same abstract equation. The proposed theory crucially depends on the existence of covariant functors among the Rota category of Galois differential algebras and suitable categories of abstract dynamical systems. In this way, a new class of integrable maps is constructed, sharing with their continuous analogs a large set of exact solutions and, in the linear case, the Picard-Vessiot group.

\end{abstract}

\date{April 25, 2015}

\maketitle

\tableofcontents

\vspace{3mm}

\section{Introduction}
\subsection{Statement of the problem}
The discretization of dynamical systems in an \textit{integrability preserving} way has been widely investigated in the last decades. Potentially, it has  a great impact in many different areas, such as discrete mathematics, algorithm theory, numerical analysis, statistical mechanics, etc. For instance, the lattice renormalization of field theories \cite{MM}, as well as several modern approaches to quantum mechanics \cite{FL} and quantum gravity \cite{ash}, \cite{thooft} require an appropriate discretization of field equations.

Discrete differential geometry and related algebraic aspects have also recently attracted much attention \cite{BS1}, \cite{BS2}, \cite{GK}, \cite{Nov1}, \cite{Nov2}, \cite{Nov3}, \cite{LNP},  \cite{Kuper1},  \cite{MV}.

The aim of this paper is to propose a general solution to the problem of integrability preserving discretization of dynamical systems with \textit{variable coefficients}. At the best of our knowledge, for this case no approach is available in full generality.



The main theorems of the paper ensure the integrability-preserving discretization of linear equations of the form
\begin{equation}
a_{N}(t) \frac{d^N}{dt^N}z + a_{N-1}(t) \frac{d^{N-1}}{dt^{N-1}}z  +\ldots+a_{1}(t) \frac{d}{dt} z+a_{0}(t) z +c_{0}(t)=0, \label{lincont}
\end{equation}
and  nonlinear equations of the form
\begin{equation}
\frac{d^m}{dt^m}z= a_{N}(t) z^{N}+a_{N-1}(t) z^{N-1}+\ldots+a_{1}(t) z+ a_{0}(t), \label{nonlincont}
\end{equation}
where $N\in\mathbb{N}$, $z:\mathbb{R}_{+} \cup \{0\}\rightarrow\mathbb{R}$ is a $\mathcal{C}^{\infty}$ function, $b_{0}(t)$ and $\{a_{0}(t),\ldots, a_{N}(t)\}$ are arbitrary polynomials in $t\in\mathbb{R}$.

In \cite{TempestaJDE}, the class of systems \eqref{nonlincont} was studied in the much simpler case $a_{0},\ldots, a_{N} \in\mathbb{R}$.

The main idea is to associate each of the systems \eqref{lincont} and \eqref{nonlincont} with an abstract equation, defined in a Galois differential algebra. Then  we realize the Galois algebra in two steps: i) by selecting a difference operator of a suitable class; ii) by endowing the algebra with a product under which the difference operator acts as a derivation.  If we represent the abstract equation as a difference equation over the Galois algebra so constructed, the discrete equation obtained shares with the continuous one some of the most relevant analytic properties.

A crucial aspect is that the discrete versions of eqs. \eqref{lincont} and \eqref{nonlincont}, obtained according to the procedure sketched above, can be interpreted as \textit{discrete analogs of integral differential equations}.  Indeed, these equations are nonlocal recurrences possessing the general form

\beq
P(\Delta) z(n)=\mathcal{E}[z(n)] \label{nonlocal}.
\eeq

Here $P(\Delta)$ is a polynomial in a discrete derivative of a given order (a delta operator), $z(n)$ is the dependent variable, $n$ the independent discrete variable, and $\mathcal{E}[z(n)]$ is a function of $z(n)$, depending on all values ranging from $n$ up to an initial point. We refer to this fact when we talk about the non-locality of the equation \eqref{nonlocal}.

The simplest case corresponds to the discretization of linear differential equations with constant coefficients. In this case, the associated discrete equations will be standard (i.e. local) difference equations, depending on a fixed number of values of $n$, determined by the order of the discrete derivative $\Delta$.

\subsection{Main results}

I) The integrable maps of the form \eqref{nonlocal}, which represent  discrete versions of the dynamical systems \eqref{lincont} or \eqref{nonlincont}, possess the following property: under suitable conditions, \textit{$\mathcal{C}^{\infty}$ solutions of the continuous systems \eqref{lincont} or \eqref{nonlincont} are converted into exact solutions of the associated integrable maps} \eqref{nonlocal}.

Precisely,  if $\sum_{k=0}^{\infty} a_{k} x^{k}$ is a $\mathcal{C}^{\infty}$ solution of an equation of the form \eqref{lincont} or \eqref{nonlincont}, then $\sum_{k=0}^{n} a_{k} p_{k}(n)$ (where $p_k(n)$ are suitable polynomials) for each $n$ will be a solution of the associated nonlocal discrete equation.

In this sense, the discrete equations \eqref{nonlocal} can also be considered as equations approximating the continuous ones up to the order $n$. Indeed, their exact solutions allow to construct  solutions of the continuous ones approximated up to the $n$-th order in their Taylor expansion.

II) The Galois theory  for the integrable maps proposed in the work is developed. In particular, it is shown that there exists an isomorphism between the \textit{Picard--Vessiot group} of a linear differential equation with constant coefficients and that of the associated difference equation obtained via our categorical approach.

\subsection{The procedure} The approach proposed combines the theory of Galois differential algebras, the theory of categories and the classical finite operator theory, in the formulation given by G.-C. Rota \cite{Rota}. Differential equations and difference analogs are interpreted as objects belonging to specific subcategories of suitable categories of dynamical systems; in this setting, the continuous and the infinitely many possible discrete versions of an equation correspond to different representations of the same abstract object. Equations representing the same object will be said to be \textit{categorically equivalent}. The main results of the paper apply indeed to categorically equivalent equations.  The point of view proposed here generalizes considerably the so-called Rota correspondence, according to the terminology used by many authors (see, for instance, \cite{LTW1} and references therein).

In \cite{TempestaJDE}, the notion of \textit{Rota differential algebra} has been introduced. It is a Galois differential algebra $\left(\mathcal{F_{\mathcal{L}}},\mathcal{Q}\right)$, where $(\mathcal{F_L}, +,\cdot, *)$ is an associative algebra of formal power series over a lattice $\mathcal{L}$ of points, endowed with an associative and commutative product "*", and $\mathcal{Q}$ is a \textit{delta operator} that acts as a \textit{derivation} with respect to this product. The idea of an adapted product, ensuring for difference operators the Leibnitz rule, was proposed in the important papers \cite{Ward}, \cite{BF}.
An analogous product has appeared, in a different context, in the theory of linear operators acting on polynomial spaces \cite{ismail}.

An intrinsic feature of the $"*"$-product is its non-locality, which is at the origin of the non-locality of the discrete equations \eqref{nonlocal}.   This property can be interpreted in terms of a no-go theorem proved in \cite{KSS}. It states that it is not possible to define a field theory on an infinite lattice endowed with a nontrivial product rule that satisfies at the same time the Leibniz rule, translational invariance and locality requirements. Besides, one should notice that a nonlocal behavior is also a very common feature in the theory of integrable difference equations.

The collection of all Rota algebras is a subcategory $\mathcal{R} (\mathcal{F},\mathcal{Q}) $ of the category of associative algebras \cite{Mac Lane}. To each choice of the variable coefficients in eqs. \eqref{lincont} and \eqref{nonlincont} it corresponds a category of linear equations $\mathcal{K}_{\{ a_{0},\ldots, a_{N},c_{0} \}}$ and of nonlinear equations $\mathcal{N}_{\{m; a_{0},\ldots, a_{N}\}}$ respectively.

Possible discretizations of a continuous equation are expressed in terms of morphisms within the same subcategory. Two functors $F: \mathcal{R} \rightarrow \mathcal{K}$ and $G: \mathcal{R} \rightarrow \mathcal{N}$  are introduced. They enable to map differential equations into difference equations by preserving the underlying differential structure. The integrability properties of a given continuous dynamical system are naturally inherited by the discrete ones associated with it and in particular, solutions of the continuous system are mapped isomorphically into solutions of the discrete ones.

An open problem is the extension of the present theory to the case of $q$-difference equations. A coherent Galois approach to this class of equations has been developed in \cite{diV}, \cite{Sauloy}. We hypothesize that a functorial correspondence can be constructed in this case, allowing a unified treatment of the Galois theory for differential equations and categorically equivalent $q$-difference equations. The discretization of partial differential equations is also a fundamental problem, that in principle can be treated with a similar categorical approach.


\section{Category theory and dynamical systems}

\subsection{Algebraic preliminaries}

The theory of delta operators, proposed in \cite{Rota} as a foundational formulation of combinatorics, has been developed in a vast body of literature (see also \cite{RR}, \cite{Roman}).
Let us denote  by $\mathcal{P}$ the space of polynomials in one variable $x\in \mathbb{K}$, where $\mathbb{K}$ is a field of characteristic zero and denote by $\mathbb{N}$ the set of nonnegative integers..  Let $\mathcal Q$ be a delta operator, and $\{p_n(x)\}_{n\in\mathbb{N}}$ be the  basic sequence of polynomials of order $n$ uniquely associated to $Q$ (see \cite{Roman}--\cite{Rota} for definitions and properties).  Let $\mathcal{F}$ be the algebra of formal power series in $x$. Since the polynomials $\{p_n(x)\}_{n\in\mathbb{N}}$ for every choice of $\mathcal Q$ provide a basis of $\mathcal{F}$, any $f\in\mathcal{F}$ can be expanded into a formal series of the form
$
f(x)=\sum_{n=0}^{\infty}a_n p_n(x) \label{exp}.
$
Let $\mathcal{L}$ be a lattice of points on the real line, isomorphic to $\mathbb{N}$. We shall denote by $\mathcal{F_L}$ the vector space of the formal power series defined on $\mathcal{L}$. The space $\mathcal{F}$ (and consequently $\mathcal{F_L}$) can be endowed with the structure of an algebra, by introducing the product defined, for each choice of $\mathcal Q$ by the relation
\begin{equation}
p_n(x)*p_m(x):=p_{n+m}(x). \label{starproduct}
\end{equation}
This product for the forward difference operator $\Delta$, has been proposed in \cite{Ward} and in \cite{ismail}. Given a choice of  $\mathcal{Q}$, one can prove that the space $(\mathcal{F}, +,\cdot, *_{\mathcal{Q}})$, endowed with the composition laws of sum of series, multiplication by a scalar and the $*$ product \eqref{starproduct} is an associative algebra.

%
%
In the following, with a slight abuse of notation we will use the symbol "$*_{\mathcal{Q}}$" whenever we wish to emphasize the dependence of the "*" product on the choice of $\mathcal{Q}$.
%
%
%
%
%
%
%

\subsection{Formal groups and delta operators} The theory of delta operators can be naturally put in the setting of formal group theory \cite{Ray}. It can be related to integrability properties of difference equations \cite{LTW1}, \cite{TempestaJDE} and number theoretical properties of zeta functions \cite{TempestaTAMS}.

Following \cite{Haze}, \cite{BMN}, let us consider the polynomial ring $\mathbb{Q}\left[ c_{1},c_{2},...\right]
$ and the formal group logarithm
$
F\left( u\right) =u+c_{1}\frac{u^{2}}{2}+c_{2}\frac{u^{3}}{3}+...
$
Let $G\left( v\right) $ be the associated inverse series, i.e. the formal group exponential
\begin{equation}
G\left( v\right) =v-c_{1}\frac{v^{2}}{2}+\left( 3c_{1}^{2}-2c_{2}\right)
\frac{v^{3}}{6}+...  \label{Ap12}
\end{equation}%
so that \textit{$F\left( G\left( v\right) \right) =v$}. The formal group law related to these series is
provided by
\begin{equation*}
\Phi \left( u_{1},u_{2}\right) =G\left( F\left( u_{1}\right) +F\left(
u_{2}\right) \right),
\end{equation*}%
and it represents the so called \textit{Lazard's Universal Formal Group
}. It is defined over the Lazard ring $L$, i.e. the subring of $%
\mathbb{Q}\left[ c_{1},c_{2},...\right] $ generated by the coefficients of
the power series $G\left( F\left( u_{1}\right) +F\left( u_{2}\right) \right)
$. In algebraic topology, delta operators are related to Thom classes and complex cobordism theory \cite{Ray2}.
A simple way to produce a family of difference delta operators is the following \cite{LTW1}. Let
\begin{equation}
\Delta_{p}=\frac{1}{\sigma}\sum_{k=l}^{m}\alpha_{k}T^{k}\text{,}\quad l\text{,
}m\in\mathbb{Z}\text{,}\mathbb{\quad}l<m\text{,\quad}m-l=p\text{,}\label{2.9}
\end{equation}
where $\sigma$ can be interpreted as a lattice spacing and $\alpha_{k}\,$ are constants  such that
\begin{equation}
\sum_{k=l}^{m}\alpha_{k}=0\text{,}\quad\sum_{k=l}^{m}k\alpha_{k}=c\text{.}\label{2.10}
\end{equation}
and $\alpha_m\neq0$, $\alpha_l\neq0$. We choose $c=1$, to reproduce possibly the derivative $D$ in the continuum limit. A difference operator of the
form (\ref{2.9}), which satisfies the equations (\ref{2.10}), is said to be a delta operator of
order $p$, if it approximates the continuous derivative up to terms of order $\sigma^p$. Apart the two conditions (\ref{2.10}), one can choose arbitrarily $m-l-1$
to fix all constants $\alpha_{k}$.
Consequently, by introducing the representation $T\sim e^{v}$,  with each delta operator of the family \eqref{2.9} one associates a realization of the one-dimensional Lazard universal formal group law. Conversely, with the identification $v\sim D$, to each formal group exponential it will correspond a delta operator (see also \cite{Ray}, \cite{TempestaTAMS}).
\subsection{The Rota category}
The notion of Rota algebra has been introduced in \cite{TempestaJDE}, as the natural Galois differential algebra over which the discretization procedure is carried out.

\begin{definition}
A Rota differential algebra is a Galois differential algebra $(\mathcal{F}, \mathcal{Q})$, where $(\mathcal{F}, +, \cdot, *_{\mathcal{Q}})$ is an associative algebra of formal power series, the product $*_{\mathcal{Q}}$ is the composition law defined by \eqref{starproduct}, and $\mathcal{Q}$ is a delta operator acting as a derivation on $\mathcal{F}$:
\begin{equation}
i) \quad \mathcal{Q}(a+b)=\mathcal{Q}(a)+\mathcal{Q}(b), \hspace{10mm} \mathcal{Q}(\lambda a)=\lambda Q(a), \quad \lambda\in\mathbb{K},
\end{equation}
\begin{equation}
ii) \quad \mathcal{Q}(a*b)= \mathcal{Q}(a)*b+a*\mathcal{Q}(b).
\end{equation}
\end{definition}

\noi As has been proved in \cite{TempestaJDE}, there exists a unique Rota algebra $(\mathcal{F}, \mathcal{Q})$ associated with a delta operator $\mathcal{Q}$.

\begin{definition}
The Rota category, denoted by $\mathcal{R}(\mathcal{F}, \mathcal{Q})$ is the collection of all Rota algebras $(\mathcal{F}, \mathcal{Q})$, with morphisms defined by
\begin{equation*}
\rho_{\mathcal{Q},\mathcal{Q}'}: \mathcal{R}(\mathcal{F}, \mathcal{Q}) \longrightarrow \mathcal{R}(\mathcal{F}, \mathcal{Q})
\end{equation*}
\begin{equation}
(\mathcal{F}, +, \cdot, *_{\mathcal Q})\longrightarrow (\mathcal{F'}, +, \cdot, *_{\mathcal{Q'}}) \label{rotamorphism}
\end{equation}
which are closed under composition.
\end{definition}
The action of the morphism $\rho_{\mathcal{Q},\mathcal{Q}'}$ on formal power series is defined by
\begin{equation}
\sum_{n} a_{n} p_{n} (x) \longrightarrow \sum_m a_{m} q_{m}(x), \label{morphism}
\end{equation}
where $\{p_{n}(x)\}_{n\in\mathbb{N}}$ and $\{q_{m}(x)\}_{m\in\mathbb{N}}$ are the basic sequences associated with $\mathcal{Q}$ and $\mathcal{Q}'$ respectively. The property of closure under composition is trivial.

\subsection{New categories for dynamical systems}
We propose two new definitions of categories for differential equations of the form \eqref{lincont} and \eqref{nonlincont}. Let us denote by $z^{*N}:=\underbrace{z*\ldots *z}_{N-times}$, with $N\in\mathbb{N}/\{0\}$.
\begin{definition}
For any choice of the set of functions $\{a_{0}(t),\ldots,a_{N}(t), c_{0}(t)\}$, where $a_{j}(t)= \sum_{m_j=0}^{M_{j}} \alpha_{j m_j} t^{m_{j}}, c_{0}(t)= \sum_{r=0}^{R} \gamma_{r}t^{r}$, with $\alpha_{j m_j}, \gamma_{r} \in\mathbb{R}$, $j=0,\ldots,N$, $N\in\mathbb{N}$, the category $\mathcal{K}_{\{ a_{0},\ldots, a_{N},c_{0}\}}$ of linear dynamical systems of order $N$ is the collection of all equations of the form
\begin{equation}
lin(\mathcal{Q},z,*_{\mathcal{Q}}):=a_{N}(t)* \mathcal{Q}^{*N} z + a_{N-1}(t)* \mathcal{Q}^{*N-1}z  +\ldots + a_{1}(t)* \mathcal{Q} z+a_{0}(t)* z +c_{0}(t)=0. \label{linabstr}
\end{equation}
\noi The set of correspondences
\beq
\lambda_{\mathcal{Q},\mathcal{Q}'}: \mathcal{K}_{\{ a_{0},\ldots, a_{N},c_{0} \}} \longrightarrow \mathcal{K}_{\{ a_{0},\ldots, a_{N},c_{0} \}},
\eeq
\beq
lin(\mathcal{Q},z,*) \longrightarrow lin(\mathcal{Q}',z,*') \label{morfeq},
\eeq
defines the class of morphisms of the category.
\end{definition}
In other words, given a differential equation of the form \eqref{lincont}, we can construct a category by considering as its objects the equation \eqref{lincont} jointly with all the infinitely many ``discretizations'' of if, obtained by varying $Q\in\mathcal{D}$.
\begin{definition}
For any choice of the set of functions $\{a_{0}(t),\ldots,a_{N}(t)\}$, where  $a_{j}(t)= \sum_{m_j=0}^{M_{j}} \alpha_{j m_j} t^{m_{j}}$, with $\alpha_{j m_j} \in\mathbb{R}$, $j=0,\ldots,N$, $N\in\mathbb{N}$, the category $\mathcal{N}_{\{m; a_{0},\ldots, a_{N}\}}$ of nonlinear dynamical systems of order $m\in\mathbb{N}$ is the collection of all equations of the form
\beq
eq(\mathcal{Q},z,*_{\mathcal{Q}}):= \mathcal{Q}^m z- a_{N}(t)*z^{*N}-a_{N-1}(t)*z^{*N-1}-\ldots-a_{1}(t)*z-a_{0}(t)=0. \label{nonlinabstr}
\eeq
The set of correspondences
\beq
\nu_{\mathcal{Q},\mathcal{Q}'}: \mathcal{N}_{\{m; a_{0},\ldots, a_{N}\}} \longrightarrow \mathcal{N}_{\{m; a_{0},\ldots, a_{N}\}},
\eeq
\beq
eq(\mathcal{Q},z,*) \longrightarrow eq(\mathcal{Q}',z,*') \label{morfeq},
\eeq
defines the class of morphisms of the category.
\end{definition}
\noi The closure of the morphisms $\lambda_{\mathcal{Q},\mathcal{Q}'}$ and $\nu_{\mathcal{Q},\mathcal{Q}'}$ under composition is easily verified.
\begin{definition}
We shall say that two objects belonging to the category $\mathcal{K}_{\{ a_{0},\ldots, a_{N},c_{0}\}}$ (or to $\mathcal{N}_{\{m; a_{0},\ldots, a_{N}\}}$) represent two categorically equivalent equations. Alternatively, two equations will be said to be categorically equivalent if there exists a morphism of categories $\lambda_{\mathcal{Q},\mathcal{Q}'}$ (or $\nu_{\mathcal{Q},\mathcal{Q}'}$) that maps them into each other.
\end{definition}
An interesting property is the structure in subcategories of the categories defined above. The inclusions are obtained by means of the natural identifications: $ \mathcal{K}_{\{ a_{0},\ldots, a_{k}\}}:= \mathcal{K}_{\{ a_{0},\ldots, a_{k},\underbrace{0,\ldots,0}_{(N-k)-times}\}}$ and $ \mathcal{N}_{\{m; a_{0},\ldots, a_{k}\}}:= \mathcal{N}_{\{ m; a_{0},\ldots, a_{k},\underbrace{0,\ldots,0}_{(N-k)-times}\}}$, for $k<N$.
\begin{lemma}
For any fixed $N\in\mathbb{N}$, and for any choice of the polynomials $\{ a_{0},\ldots, a_{N},c_{0}\}$, there exists a filtration of subcategories, given by the finite sequences
\beq
\mathcal{K}_{\{ a_{0}\}}\subset\mathcal{K}_{\{ a_{0},a_{1}\}}\subset\ldots\subset\mathcal{K}_{\{ a_{0},\ldots, a_{k}}\}\subset\ldots\subset\mathcal{K}_{\{ a_{0},\ldots, a_{N},c_{0}\}} \label{filtrF}
\eeq
and
\beq
\mathcal{N}_{\{m; a_{0}\}}\subset\mathcal{N}_{\{m;  a_{0},a_{1}\}}\subset\ldots\subset\mathcal{N}_{\{m; a_{0},\ldots, a_{k}}\}\subset\ldots\subset\mathcal{N}_{\{m; a_{0},\ldots, a_{N}\}}  \label{filtrG}
\eeq
respectively.
\end{lemma}
\begin{proof}
For each of the subsets $\mathcal{K}_{\{ a_{0},\ldots, a_{k}}\}$, the restriction of the morphisms $\lambda_{\mathcal{Q},\mathcal{Q}'}$ over these subsets still preserves the closure under composition and defines morphisms of subcategories. The same argument holds for the sequence related to $\mathcal{N}_{\{a_{0},\ldots, a_{N}\}}$.
\end{proof}

As a consequence of the previous construction, we can define functors relating the category of Rota differential algebras with those of abstract dynamical systems defined above.
\begin{theorem} \label{functor}
For any choice of the functions $\{ a_{0}(t),\ldots, a_{N}(t),c_{0}(t)\}$, the application
\beq
F: \mathcal{R}(\mathcal{F},\mathcal{Q}) \longrightarrow \mathcal{K}_{\{ a_{0},\ldots, a_{N},c_{0} \}}, \label{functorF}
\eeq
\[
(\mathcal{F}+, \cdot, *_{\mathcal{Q}}) \longrightarrow lin(\mathcal{Q},z,*),
\]
\[
\rho_{\mathcal{Q},\mathcal{Q}'} \longrightarrow \lambda_{\mathcal{Q},\mathcal{Q}'},
\]
is a covariant functor.
\end{theorem}
\begin{proof}
A direct verification shows that $F$ preserves the composition of morphisms:
\[
F(\lambda_{\mathcal{Q}'',\mathcal{Q}'}\circ \lambda_{\mathcal{Q}',\mathcal{Q}})=F(\lambda_{\mathcal{Q}'',\mathcal{Q}'})\circ F(\lambda_{\mathcal{Q}',\mathcal{Q}}).
\]
If we denote by $id_{\mathcal{Q}}:=\lambda_{\mathcal{Q},\mathcal{Q}}$ the identity morphism, we also have
\[
F(id_{\mathcal{Q}}(\mathcal{A}))=id_{\mathcal{Q}}(F(\mathcal{A})),
\]
where $\mathcal{A}\in \mathcal{R}(\mathcal{F}, \mathcal{Q})$.
\end{proof}
\noi In the same manner one can prove that the application
\beq
G: \mathcal{R}(\mathcal{F},\mathcal{Q}) \longrightarrow \mathcal{N}_{\{ m; a_{0},\ldots, a_{N} \}}, \label{functorG}
\eeq
\[
(\mathcal{F}+, \cdot, *_{\mathcal{Q}}) \longrightarrow eq(\mathcal{Q},z,*),
\]
\[
\rho_{\mathcal{Q},\mathcal{Q}'} \longrightarrow \nu_{\mathcal{Q},\mathcal{Q}'},
\]
is a covariant functor.

The functors \eqref{functorF} and \eqref{functorG} encode the main features of the discretization procedure we propose, and provide \textit{functorial Rota correspondences} among continuous and discrete dynamical systems. These functors can also be defined on the subcategories defined above.
\begin{corollary}
The restriction of the functors $F$ and $G$ to the subcategories defined by the filtrations \eqref{filtrF} and \eqref{filtrG} respectively keep being covariant functors.
\end{corollary}
\section{Main theorems}

In this section, we present the main results of the paper. Theorem \ref{main1} offers a solution to the problem of the integrability preserving discretization of a linear $n$-th order ODE with variable coefficients of polynomial type. Theorem \ref{main2} solves the analogous problem for the case of a nonlinear first-order ODE.


\subsection{Integrable maps from linear ODEs}
\begin{theorem} \label{main1}
\noindent Consider the differential equation
\begin{equation}
lin(\partial, z):=a_{N}(t) \frac{d^N}{dt^N}z + a_{N-1}(t) \frac{d^{N-1}}{dt^{N-1}}z  +\ldots+a_{1}(t) \frac{d}{dt} z+a_{0}(t) z +c_{0}(t)=0, \label{lincont2}
\end{equation}
where $z\in\mathcal{C}^{\infty}(\mathbb{R}_{+} \cup \{0\},\mathbb{R})$, and suppose that $a_{j}(t)= \sum_{m_j=0}^{M_{j}} \alpha_{j m_j} t^{m_{j}}$, $c_{0}(t)= \sum_{r=0}^{R} \gamma_{r}t^{r}$, with $\alpha_{j m_j}, \gamma_{r} \in\mathbb{R}$, $M_{j}, R\in\mathbb{N}$, $j=0,\ldots,N$. Assume that
\beq
z(t)=\sum_{k=0}^{\infty} b_k t^k
\eeq
be a real solution of \eqref{lincont2} in the ring of formal power series in $t \in \mathbb{R}_{+} \cup \{0\}$. Let $k,j,m\in\mathbb{N}$, $k\geq j+m$ and
\beq
K(k,j,m,n):=(-1)^{k-j-m}\frac{1}{j!(k-m-j)!}\frac{n!}{(n-k)!}.
\eeq
\noi Then the equation
\begin{eqnarray}
\nn \sum_{l=0}^{N}\left[a_{l0} \Delta^{l} z_{n} + \sum_{m_l=1}^{M_l} \alpha_{l m_{l}}\sum_{k=0}^{n} \sideset{}{'}\sum_{j_l=0}^{k-m_l}K(k, j_l, m_l,n) \Delta^{l} z_{j_l}\right] + \sum_{r=0}^{R} \beta_{r}\frac{R!}{(R-r)!}=0 \\ \label{eqlin}
\end{eqnarray}
\noi that represents eq. \eqref{lincont2} on a regular lattice of points $\mathcal{L}$ indexed by the variable $n\in\mathbb{N}$, admits as a solution the series
\beq
z_{n}=\sum_{k=0}^{n} b_k \frac{n!}{(n-k)!}. \label{part1}
\eeq
Here we denote by $\sideset{}{'}\sum$ a sum ranging over all values of the indices $k_{i}$ such that the sum is not empty.
\end{theorem}

\begin{proof}
We introduce an auxiliary space of variables, defined in terms of the Fourier coefficients of the expansion of $z$ in terms of the basic polynomials $\{p_k\}_{k\in\mathbb{N}}$ for a given delta operator $\mathcal{Q}$. The transformation
\begin{equation}
z(t)= \sum_{k=0}^{\infty}\zeta_{k} p_{k} (t), \hspace{10mm}  \label{interpol}
\end{equation}
 is said to be a \textit{discrete interpolating transformation} with Fourier coefficients $\zeta_{k} \in\mathbb{R}$.

Therefore, we shall identify $z(t)$ with the unique associated sequence of its Fourier coefficients $\{\zeta_{n}\}_{n\in\mathbb{N}}$.
\noindent The discrete transform (\ref{interpol}) is finite whenever $t\in\mathcal{L}$, where the lattice $\mathcal{L}$ is chosen in such a  way that its points coincide with the zeros of the sequence of polynomials $\{p_{k}(t)\}_{k\in\mathbb{N}}$. Hereafter we will choose $\mathcal{L}$ to be an equally spaced lattice, indexed by $n\in\mathbb{N}$. In this case, the polynomials guaranteeing the finiteness of expansion \eqref{interpol} are lower factorial polynomials, which in turn are the basic polynomials for the operator $\Delta$. Indeed, we have
\begin{equation}
p_{k}(n)= \begin{cases} 0 \qquad\qquad\qquad if \quad n<k, \\
\frac{n!}{(n-k)!}\qquad\qquad if\quad n \geq k. \label{pol}
\end{cases}
\end{equation}
\noindent Let us introduce the function $z:\mathbb{N}\rightarrow \mathbb{R}$ defined by
\begin{equation}
z_n=  \sum_{l=0}^{n} \frac{n!}{(n-l)!}\zeta_l.
\end{equation}
It corresponds to the restriction on $\mathcal{L}$ of the expansion \eqref{interpol} of $z(t)$. For its \textit{inverse interpolating transform}, we have
\begin{equation}
\zeta_n= \sum_{l=0}^{n} (-1)^{n-l} \frac{1}{  l!(n-l)!}z_l. \label{invtr}
\end{equation}
For the purpose of constructing the equation equivalent to eq. \eqref{lincont2} in $\mathcal{K}_{\{ a_{0},\ldots, a_{N},c_{0}\}}$, we represent on the lattice $\mathcal{L}$ the product $t z(t)$ by means of the action of the morphism $\rho_{\partial,\Delta}: (\mathcal{F},+,\cdot)\rightarrow (\mathcal{F},+,*_{\Delta})$ between the standard Rota algebra of $\mathcal{C}^{\infty}$ functions with the pointwise product and that associated with the forward difference operator. We get the contribution
\beqa
\nn \sum_{k=0}^{\infty}b_{k}t^{k+1}&\longrightarrow&  \sum_{k=1}^{n} \frac{n!}{(n-k-1)!}\zeta_k=\sum_{k=0}^{n} \sum_{j=0}^{k} (-1)^{k-j} \frac{1}{  j!(k-j)!}\frac{n!}{(n-k-1)!}z_j
\\&=&\sum_{k=0}^{n} \sum_{j=0}^{k-1} (-1)^{k-j-1} \frac{1}{  j!(k-j-1)!}\frac{n!}{(n-k)!}z_j.
\eeqa
At the same time, $\frac{d^p}{dt^p}z \longrightarrow \Delta^{p}z$, with $p\in\mathbb{N}/\{0\}$. This implies that
\beq
\nn t \frac{dz}{dt} \overset{\rho_{\partial,\Delta}}{\longrightarrow} \sum_{k=0}^{n} \sum_{j=0}^{k-1} (-1)^{k-j-1} \frac{1}{  j!(k-j-1)!}\frac{n!}{(n-k)!}(z_{j+1}-z_{j}).
\eeq
Similar expressions hold for higher-order derivatives. By iterating the previous reasoning to all contributions of the polynomial form $\sum_{m_j=0}^{M_{j}} \alpha_{j m_j} t^{m_{j}}\frac{d^p}{dt^p}z$, we prove that eq. \eqref{eqlin} is categorically equivalent to \eqref{lincont2}. Both in turn are representations of the abstract equation \eqref{linabstr}.

To conclude the proof, observe that, by means of the action of the functor $F$, any $\mathcal{C}^\infty$ solution (and in particular analytic) $z$ of eq. \eqref{lincont2} is carried into a solution $x_n$ of eq. \[lin(Q,z,*_{Q})=\lambda_{\partial,\Delta}(lin(\partial, z,\cdot)).\]
In addition, on the lattice $\mathcal{L}$, the sum $\sum_{k} b_k p_{k}(n)$ truncates and converts into the finite sum \eqref{part1}. \end{proof}
As a consequence of this proof,  eqs. \eqref{lincont2} and \eqref{eqlin} are categorically equivalent.
\begin{remark}\label{rem10}
Theorem \ref{main1} can be considerably generalized by choosing an arbitrary different delta operator $Q$ and the corresponding basic sequence $\{q_{n}(x)\}_{n\in\mathbb{N}}$. In this case, we define the lattice $\mathcal{L}$ to be the union set of all zeroes of the polynomials of the sequence $\{q_{n}(x)\}_{n\in\mathbb{N}}$. The action of the morphism $\lambda_{\partial,Q}$ will provide a different representation of eq. \eqref{lincont2}, categorically equivalent to it in $\mathcal{K}_{\{ a_{0},\ldots, a_{N},b_{0}\}}$.
\end{remark}

\subsection{Integrable maps from nonlinear ODEs}

\begin{theorem} \label{main2}
\noindent Consider a dynamical system of the form
\begin{equation}
eq(\partial, z):=\frac{d^m}{dt^m}z- a_{N}(t) z^{N}-a_{N-1}(t) z^{N-1}-\ldots+a_{1}(t) z- a_{0}(t)=0, \label{ncont2}
\end{equation}
where $z\in\mathcal{C}^{\infty}(\mathbb{R}_{+} \cup \{0\},\mathbb{R})$, $m\in\mathbb{N}$, and suppose that $a_{j}(t)= \sum_{m_j=0}^{M_{j}}\alpha_{j m_j} t^{m_{j}}$,  with $\alpha_{j m_j}\in\mathbb{R}$,  $j=0,\ldots,N$.
Assume that
\beq
z(t)=\sum_{k=0}^{\infty} b_k t^k
\eeq
be a real solution of \eqref{ncont2} in the ring of formal power series in $t \in \mathbb{R}_{+} \cup \{0\}$. Let
\beq
M(k_1,\ldots,k_{j},n):=\frac{(-1)^{k_1+\cdots + k_{j}+n}}{k_1!\ldots k_{j}!}\frac{ (j-1)^{n-k_{1}-k_{2}-\cdots-k_{j}}}{(n-k_1-k_2-\ldots-k_{j})!}
\eeq
Then the difference equation

\begin{eqnarray}
\nn \Delta^{m}z_{n}&=&z_{n+m}-\binom{m}{1}z_{m+m-1}+\binom{m}{2}z_{n+m-2}+\ldots+ z_{n} \\ \nn  &=&n!\sum_{j=1}^{N} \left[  \sum_{m_j=0}^{M_{j}} \alpha_{j m_j} t^{m_{j}} \sideset{}{'}\sum_{k_1,\ldots,k_{j}=0}^{n} M(k_1,\ldots,k_{j},n-m_{j}) z_{k_1}z_{k_2}\cdots z_{k_j}\right]  + \sum_{r=0}^{R} \gamma_{r}\frac{n!}{(n-r)!}, \\ \label{eqnonlin}
\end{eqnarray}
that represents eq. \eqref{ncont2} on a regular lattice of points $\mathcal{L}$ indexed by the variable $n\in\mathbb{N}$, admits as a solution the series
\beq
z_{n}=\sum_{k=0}^{n} b_k \frac{n!}{(n-k)!}. \label{part}
\eeq
Here we denote by $\sideset{}{'}\sum$ a sum ranging over all values of the indices $k_{i}$ such that $\sum_{i} k_{i}\leq n$.
\end{theorem}

\begin{proof}

Let us denote by $z:\mathbb{N}\rightarrow \mathbb{R}$ the restriction of $z$ on $\mathcal{L}$. First, we need  to compute explicitly the product $z^{*p}_{n}:=(\underbrace{z*\ldots *z}_{p-times})(n)$.

\noi We get
\beqa
\nn z^{*p}_{n}&=&\sum_{l_1,l_2,\ldots l_{p}=0}^{\infty} \zeta_{l_1}\zeta_{l_2}\cdots \zeta_{l_{p}} P_{l_1+l_2+\ldots +l_{p}}(n)\\ \nn &=&
\sum_{l_1,\ldots l_p=0}^{n}\ \sum_{k_1=0}^{l_1}  \ldots\sum_{k_p=0}^{l_p}(-1)^{l_1-k_1+l_2-k_2+\ldots+l_p-k_p}[ \frac{z_{k_1}z_{k_2}\cdots z_{k_p}}{k_1!(l_1-k_1)!k_2!(l_2-k_2)! \cdots k_p! (l_p-k_p)!} \\ \nn
&\cdot& \frac{n!}{(n-l_1-\ldots-l_p)!}]=\sum_{k_1,\ldots k_p=0}^{n} \frac{(-1)^{k_1+\ldots + k_p}}{k_1!\cdots k_p!}z_{k_1}z_{k_2}\cdots z_{k_p} \mathrm{K}_{n,{k_1},\ldots,{k_p}},
\eeqa
\noindent where we have introduced the kernel
\beqa
\nn \mathrm{K}_{n,{k_1},\ldots,{k_p}}:= \sideset{}{'}\sum_{l_1, \ldots,l_p=0}^{n}(-1)^{l_1+\ldots+l_p}\frac{1}{(l_1-k_1)!} \cdots \frac{1}{(l_p-k_p)!}\frac{n!}{(n-l_1-\ldots -l_p)!}. \\ \label{kernel}
\eeqa
This expression can be rewritten as
\beq
\mathrm{K}_{n,{k_1},\ldots,{k_p}}=\sum_{l=0}^{n}\frac{(-1)^{l}n!}{(n-l)!}\sum_{l_{1}=k_{1}}^{l}\frac{1}{(l_1-k_1)!}\frac{2^{l-l_1-k_2-\ldots-k_p}}{(l-l_1-k_2-\ldots-k_p)!},
\eeq
where $l=l_1+\ldots+l_p$. After some algebraic manipulations, it reduces to
\beqa
\nn \mathrm{K}_{n,{k_1},\ldots,{k_p}}&=&\sum_{s_1=0}^{n-k_1}\ldots\sum_{s_{p-2}=0}^{n-k_{p-2}}\frac{1}{s_1!\ldots s_{p-2}!}\frac{(-1)^{n}n!}{\left(n-s_1-\ldots-s_{p-2}-k_{1}-\ldots-k_{p}\right)!}=\\ \nn
&=& (-1)^{n}n! \sum_{s_1=0}^{n-k_1}\frac{1}{s_1!}\ldots\sum_{s_{p-3}=0}^{n-k_{p-3}}\frac{1}{s_{p-3}!}\sum_{s_{p-2}=0}^{n-s_1-\ldots-s_{p-3}-k_{1}-\ldots-k_{p}}\frac{1}{s_{p-2}!}
\\ \nn &\cdot&\frac{1}{(n-s_1-\ldots-s_{p-2}-k_{1}-\ldots-k_{p})!}=
\\ &=&\nn (-1)^{n}n! \sum_{s_1=0}^{n-k_1}\frac{1}{s_1!}\ldots\sum_{s_{p-3}=0}^{n-k_{p-3}}\frac{1}{s_{p-3}!}\frac{2^{n-s_1-\ldots-s_{p-3}-k_{1}-\ldots-k_{p}}}{(n-s_1-\ldots-s_{p-3}-k_{1}-\ldots-k_{p})!}=
\\ \nn &=&(-1)^{n}n! \sum_{s_1=0}^{n-k_1}\frac{1}{s_1!}\ldots\sum_{s_{p-4}=0}^{n-k_{p-4}}\frac{1}{s_{p-4}!}\frac{3^{n-s_1-\ldots-s_{p-4}-k_{1}-\ldots-k_{p}}}{(n-s_1-\ldots-s_{p-4}-k_{1}-\ldots-k_{p})!}=\ldots
\eeqa
By iterating the summation procedure, we arrive at the final expression
\begin{equation}
\mathrm{K}_{n,{k_1},\ldots,{k_p}}= \frac{(-1)^{n}n! (p-1)^{n-k_1-\ldots-k_p}}{(n-k_1-\ldots-k_p)!}, \qquad n>k_{1}+\ldots+k_{p}. \label{39}
\end{equation}
We deduce
\beq
 z^{*p}_{n}=n!\sum_{k_1,\ldots k_p=0}^{n} \frac{(-1)^{k_1+\ldots + k_p+n}}{k_1!\cdots k_p!} \frac{ (p-1)^{n-k_1-\ldots-k_p}}{(n-k_1-\ldots-k_p)!} z_{k_1}z_{k_2}\cdots z_{k_p}, \qquad n>k_{1}+\ldots+k_{p}.
\eeq

It is not difficult to ascertain that the representation of the product $t^{*m_j}*z^{*p}$  has an expression very similar to eq. \eqref{39}, with the variable $n$ replaced by $n-m_{j}$.

By combining together all the previous results,
we obtain the proof that the difference equation \eqref{eqnonlin} is categorically equivalent to eq. \eqref{ncont2} on $\mathcal{L}$, i.e. the image of eq. \eqref{ncont2} under the action of the morphism $\nu_{\partial,\Delta^{+}}$ defined in $\mathcal{N}_{\{m; a_{0},\ldots, a_{N}\}}$. In turn, both are realizations of eq. \eqref{nonlinabstr}.

To prove that the series \eqref{part} is solution of eq. \eqref{eqnonlin}, observe that the morphism $\rho_{\partial,\Delta^{+}}$ provides the correspondence

\begin{equation}
\sum_{k} b_k t^k \longrightarrow \sum_{k} b_k p_{k}(n).
\end{equation}

The action of the functor $G$ will carry any $\mathcal{C}^\infty$ solution $z$ of eq. \eqref{ncont2},  defined on the algebra $(\mathcal{F}+,\cdot)$, into a $\mathcal{C}^\infty$   solution $z_n$ of the corresponding equation
\[eq(\Delta,z,*_{\Delta})=\nu_{\partial,\Delta}(eq(\partial, z,\cdot)),\] defined on the Rota algebra $(\mathcal{F},+,*_{\Delta})$. Once represented this equation on the lattice $\mathcal{L}$, the series expansion of the solution $z_n$ truncates and converts into the finite sum \eqref{part}.
\end{proof}


%

The same scheme of Remark \ref{rem10} can be applied to generalize Theorem \ref{main2} to arbitrary objects of $\mathcal{N}_{\{m; a_{0},\ldots, a_{N}\}}$.
\section{Integrable dynamics on the Fourier space}
The integrable maps obtained by means of the previous construction define an auxiliary dynamics in the space of their Fourier coefficients, that is of independent interest. In this section, we shall focus on the case of nonlinear equations with constant coefficients, which provides a neat example of this alternative construction.
\begin{proposition}
\noindent Consider a dynamical system of the form
\begin{equation}
\frac{d^m}{dt^m}z= a_{N} z^{N}+a_{N-1} z^{N-1}+\ldots+a_{1} z+ b_{0}, \label{ncont3}
\end{equation}
where $z\in\mathcal{C}^{\infty}(\mathbb{R}_{+} \cup \{0\},\mathbb{R})$, $m\in\mathbb{N}$, and $a_{1},\ldots,a_{N}, b_{0} \in\mathbb{R}$. Assume that
\beq
z=\sum_{k=0}^{\infty} b_k t^k
\eeq
be a real solution of \eqref{ncont3} in the ring of formal power series in $t \in \mathbb{R}_{+} \cup \{0\}$.
Then the map
\beqa
\nn \frac{(n+m)!}{n!}\zeta_{n+m}&=&a_{N}\sum_{\overset{l_1,\ldots,l_{N-1}=0}{l_1+\ldots + l_{N-1}\leq n}}\zeta_{l_1}  \cdots \zeta_{l_{N-2}} \zeta_{l-l_{1}-\ldots-l_{N-1}} +\ldots \\
&+&a_{2} \sum_{l_1=0}^{n}\zeta_{l_1}  \zeta_{n-l_{1}}+a_{1} \zeta_{n} +b_{0}
\eeqa
possesses the solution
\beq
\zeta_{n}=\sum_{l=0}^{n}\sum_{k=0}^{l}\frac{(-1)^{n-l}b_{k}}{(n-l)!(l-k)!}.
\eeq
\end{proposition}
\begin{proof}
As a consequence of the definition of basic sequence for $\{p_{k}\}_{k\in\mathbb{N}}$, one can easily prove that, for each $m\in\mathbb{N}/\{0\}$
\beq
\Delta^{m} z_{n}=\sum_{l=0}^{n}\frac{n!}{(n-l)!} \frac{(l+m)!}{l!}\zeta_{l+m}.
\eeq
Also,
\beq
z_n^{*p}=\sum_{l_1,\ldots,l_p=0}^{\infty}\zeta_{l_1}\ldots\zeta_{l_p}p_{(l_{1}+\ldots+l_{p})}(n)=\sum_{l=0}^{n}\frac{n!}{(n-l)!}\cdot\sum_{\overset{l_1,\ldots,l_{p-1}=0}{l_1+\ldots + l_{p-1}\leq l}}\zeta_{l_1}  \cdots \zeta_{l_{p-2}} \zeta_{l-l_{1}-\ldots-l_{p-1}}
\eeq
By combining the previous expressions, the thesis follows.
\end{proof}
The linear case, as well as that of dynamical systems with nonconstant coefficients can be treated in the same way. The study of these cases is left to the reader.
\section{Linear difference equations over Galois algebras and Picard-Vessiot theory}
The aim of this section is to extend some relevant results of Galois theory, established  for homogeneous linear differential equations with constant coefficients, to their \textit{alter ego} on $\mathcal{L}$ defined by Theorem \ref{main1} (See e.g. \cite{Kolchin} and \cite{PS2003} for definitions and results).

Let $\{a_{0},\ldots,a_{N-1}\}\in\mathbb{K}$. Let $\mathcal{B(\mathcal{F})}$ be the space of linear operators acting on the Rota differential algebra   $(\mathcal{F}+, \cdot, *_{\mathcal{Q}})$ over $\mathbb{K}$. We introduce the linear operator $T[Q]\in\mathcal{B(\mathcal{F})}$ defined by $T[Q]:= Q^{N}+a_{N-1} Q^{N-1}+\ldots +a_{1} Q+a_{0}$. We consider the linear differential equation
\beq
T[{\partial}](z):= y^{(N)}+a_{N-1} y^{(N-1)}+\ldots +a_{1}y'+a_{0} y =0,  \label{diff}
\eeq
and its categorically equivalent difference equation
\beq
T[{\Delta}](z):=\Delta^{N}z+a_{N-1} \Delta^{N-1}z+\ldots +a_{1} \Delta z +a_{0} z=0. \label{discr}
\eeq
It is trivial to show that the rings $C_{\partial}$ and $C_{\Delta}$ of constants for $\mathcal{R}(\mathcal{F}, \Delta)$ and $\mathcal{R}(\mathcal{F}, \Delta)$ coincide.
The following result holds.
\begin{lemma} \label{fundamental}
The morphism $\rho_{\partial,\Delta}: \mathcal{R}(\mathcal{F}, \mathcal{\partial}) \longrightarrow \mathcal{R}(\mathcal{F}, \Delta)$ maps isomorphically a fundamental system of solutions $\mathcal{S}$ of the linear differential equation \eqref{diff} into a fundamental system $\rho_{\partial, \Delta}(\mathcal{S})$ of the linear difference equation \eqref{discr}.
\end{lemma}
\begin{proof}
As a consequence of Theorem \ref{main1}, $\mathcal{C}^{\infty}$ solutions of  eqs. \eqref{diff} and \eqref{discr} are in one-to one correspondence. The morphism $\rho_{\partial,\Delta}$ maps basic sequences into basic sequences, so it preserves linear independence of solutions and the dimension of the associated vector space. Then $\mathcal{S}'=\rho_{\partial, \Delta}(\mathcal{S})$  is a fundamental set for eq. \eqref{discr}.
\end{proof}
We wish to define a Picard--Vessiot extension for eq. \eqref{discr} by using the categorical approach developed before. To this aim, first we construct the \textit{universal solution algebra} $\mathcal{U}$ of eq. \eqref{discr}. Let
\beq
\Delta Y =AY,\qquad A\in\mathfrak{gl}_{N}(\mathcal{F}) \label{matrixdiff}
\eeq
be its matrix form, where it is understood that $Y_{i+1,j}=\Delta Y_{i,j}$. We introduce an $N\times N$ matrix of indeterminates $Y=(Y_{i,j})$ and define
\beq
\mathcal{U}[\Delta]:=\mathcal{F}\left[Y_{ij},\frac{1}{det(Y_{i,j})}\right], \quad 1\leq i,j\leq N.
\eeq
Here we introduce the modified Wronskian
\beq\label{wronsk}
det(Y_{i,j}):=
\begin{bmatrix}
Y_{11} &Y_{12}&\ldots& Y_{1N}\\
\Delta Y_{11}&\Delta Y_{12}&\ldots &\Delta Y_{1N}\\
\Delta^2 Y_{11}&\ldots& & \ldots\\
\vdots & & &\vdots \\
\Delta^{N-1}Y_{11}&\ldots& &\Delta^{N-1} Y_{1N}
\end{bmatrix} \ .
\eeq
The columns of $Y$ are formed by a fundamental system of solutions of the matrix equation $\Delta Y =AY$.
\begin{definition}
A Picard--Vessiot (PV) ring for the matrix equation \eqref{matrixdiff} over the Rota algebra $\mathcal{R}(\mathcal{F},\Delta)$ is a simple differential ring $\mathcal{O}$ such that

i) There exists a fundamental matrix $Z\in \text{GL}_{n}(R)$ for eq. \eqref{matrixdiff}.

ii) The ring $\mathcal{O}$  is generated by $\mathcal{F}$, the entries of $Z$ and $\frac{1}{det(Z)}$.
\end{definition}
To construct a PV-ring, we introduce the notion of \textit{difference ideal}.
\begin{definition}
A difference ideal is an ideal generated by the elements of the form
\beq
 \Delta^{N} X_{j} +a_{N-1} \Delta^{N-1} X_{j}+\ldots +a_{1} \Delta X_{j} +a_{0} X_{j}, \qquad 1\leq j\leq N,
\eeq
and by those obtained from them through the application of the operator $\Delta$.
\end{definition}
The following result holds.
\begin{proposition}
Let $\mathcal{I}$ be a maximal difference ideal of $\mathcal{U}[\Delta]$. Then the ring $\mathcal{V}[\Delta]:=\mathcal{U}[\Delta]/\mathcal{I}$ is a \textit{Picard-Vessiot ring} for the equation \eqref{matrixdiff}.
\end{proposition}
\begin{proof}
It suffices to observe that in $\mathcal{U}[\Delta]/\mathcal{I}$ the Wronskian \eqref{wronsk} is invertible.
\end{proof}
With an abuse of notation, we shall denote by $\mathcal{U}[\partial]$ the universal solution algebra of the matrix differential equation $Y'=AY$ associated to eq. \eqref{diff}, where $Y_{i+1,j}= Y'_{i,j}$.
We introduce the notion of a differential Galois group for the previous equations.
\begin{definition}
Given the Rota algebra $\mathcal{R}(\mathcal{F},\Delta)$, let $\mathcal{V}[\Delta]$ be a Picard-Vessiot ring over $\mathcal{R}$. The differential Galois group of $\mathcal{V}[\Delta]$ over $\mathcal{R}$, $\text{DGal}(\mathcal{V}/ \mathcal{R})$ is the group of the differential $\mathcal{R}$-isomorphisms.
\end{definition}
Let $\mathcal{M}$ denote a maximal differential ideal of eq. \eqref{diff}. The main result of this section is the following
\begin{theorem}
Let $\mathcal{W}[\partial]:=\mathcal{U}[\partial]/\mathcal{M}$ a PV-ring for the linear homogeneous differential equation \eqref{diff} and $\mathcal{V}[\Delta]=\mathcal{U}[\Delta]/\mathcal{I}$ a PV-ring for the linear homogeneous difference equation \eqref{discr}. Then there exists an isomorphism of groups $\Phi$ such that
\beq
\text{DGal}(\mathcal{W}/ \mathcal{R(\mathcal{F}, \partial)})\stackrel{\Phi}{\longleftrightarrow} \text{DGal}(\mathcal{U}[\Delta]/ \mathcal{R}(\mathcal{F}, \Delta)).
\eeq
\end{theorem}
\begin{proof}
According to Lemma \ref{fundamental}, given a fundamental set $\mathcal{S}$ of solutions of eq. \eqref{diff}, $\mathcal{S}'=\rho_{\partial, \Delta}(\mathcal{S})$ is a fundamental set of solutions of eq. \eqref{discr}. This implies that the universal solution algebras $\mathcal{U}[\partial]$ and $\mathcal{U}[\Delta]$ are isomorphic.  The morphism $\lambda_{\partial, \Delta}$ maps a (maximal) differential ideal $\mathcal{M}$ into a (maximal) difference ideal $\mathcal{I}$, so that the associated Picard-Vessiot rings are isomorphic. Then the differential isomorphisms of the Galois groups associated with eqs. \eqref{diff} and \eqref{discr} are in one-to one correspondence.
\end{proof}
\begin{remark}
A completely analogous result is valid for the case of the difference equation $T[\Delta^{-}](z)=0$, where $\Delta^{-}:=1-T^{-1}$. However, the previous construction does not necessarily hold for a generic delta operator $Q$. Indeed, the fundamental set of eq. $T[Q](z)=0$ generates a linear space whose dimension is in general greater than $N$.
\end{remark}

\section{Nonlocal Lie symmetries and integrable maps: some conjectures}

The categorical discretization proposed in Theorems \ref{main1} and \ref{main2} allows to map isomorphically $\mathcal{C}$$^\infty$ solutions of differential equations into solutions of difference equations. Consequently, we hypothesize the existence of a correspondence among the Lie symmetry groups and algebras admitted by the continuous and discrete models respectively. For an account of the modern theory of Lie symmetries, see \cite{olver}.

The first conjecture we formulate concerns the existence of symmetry transformations.
\begin{conjecture}
The integrable map \eqref{eqlin}  (resp. \eqref{eqnonlin}) admits a Lie group $\mathcal{G}$ of nonlocal diffeomorphisms that leave the map invariant and transform solutions into solutions.
\end{conjecture}
Objects categorically equivalent usually do not possess isomorphic fundamental sets of solutions, with the exception of linear equations with constant coefficients, treated in Section 6. Consequently, the full Lie algebras generated by the symmetries postulated in the previous conjecture, in general are not isomorphic. The following conjecture proposes a weaker connection among these algebras.


\begin{conjecture}
The Lie algebra of the generators of the Lie group $\mathcal{G}$ of the nonlocal symmetry diffeomorphisms associated with the map \eqref{eqlin} (resp. \eqref{eqnonlin}) contains a subalgebra which is isomorphic to the Lie algebra of the classical Lie point symmetries of the continuous dynamical system \eqref{lincont} (resp.  \eqref{nonlincont}).
\end{conjecture}
\appendix

\section{Some applications of the main theorems}

Here we propose some examples of integrable maps associated with ODEs relevant in the applications.

\subsection{A discrete classical Harmonic Oscillator}
The equation for the classical harmonic oscillator
\beq
y''(x)+\omega^2 y(x)=0
\eeq
possesses the simple discrete analog
\beq
z_{n+2}-2 z_{n+1}+(\omega^2 +1) z_{n}=0.
\eeq
It admits the two independent solutions
\beq
z^{(1)}:= \sum_{k=0}^{n} s_{k} \frac{n!}{(n-k)!}, \hspace{10mm} z^{(2)}:=\sum_{k=0}^{n} c_{k} \frac{n!}{(n-k)!}
\eeq
where
\beqa
\nn s_{k}&=&\frac{1}{k!}\frac{d}{dx^k}[\sin (x)]\mid_{x=0}, \\
c_{k}&=&\frac{1}{k!}\frac{d}{dx^k}[\cos (x)]\mid_{x=0}.
\eeqa
The damped harmonic oscillator
\beq
y''(x)+2 q \omega y'(x)+\omega^2 y(x)=0
\eeq
in the case $q\leq 1$ admits the general solution
\beq
y_{gen}(x)=A e^{-q \omega x}\sin \left( \sqrt{1-q^2} \omega x+ \phi\right).
\eeq
The \textit{discrete damped oscillator} reads
\beq
z_{n+2} +2(q \omega -1) z_{n+1} +(\omega^2-2q\omega+1)z_{n},
\eeq
with the general solution
\beq
z_{n}=\sum_{k=0}^{n}g(k)\frac{n!}{(n-k)!}, \hspace{10mm} g(k):=\frac{1}{k!}\frac{d}{dx^k}[y_{gen}]\mid_{x=0}.
\eeq
\subsection{A difference equation for the Gaussian function}
\noi As a first example, we shall consider the differential equation for the Gaussian function:
\beq
y'(x)=-x y(x).
\eeq
According to Theorem \ref{main1}, the associated integrable map reads
\beq
z_{n+1}-z_{n}+\sum_{k=0}^{n}\sum_{j=0}^{k-1}(-1)^{k-1-j}\frac{1}{j! (k-1-j)!}\frac{n!}{(n-k)!}z_{j}=0.
\eeq
It admits the solution
\beq
z_{n}=\sum_{k=0}^{n}g(k)\frac{n!}{(n-k)!},
\eeq
where $g(k)$ is the $k$-th coefficient of the Taylor expansion of the Gaussian function
\beq
e^{-\frac{x^2}{2}}=\sum_{k=0}^{\infty}g(k)\frac{x^k}{k!}.
\eeq
\subsection{Hypergeometric equation}

The real hypergeometric differential equation
\beq
t(1-t)\frac{d^2 z}{dt^2}+\left[c-(a+b+1)\omega\right]\frac{dz}{dt}-ab z=0 \label{hde}
\eeq
where $a,b,c\in\mathbb{R}$ admits three regular singular points $\{0,1, \infty\}$ \cite{AS}. By way of example, we shall restrict to the singularity at $t=0$. Around this point, eq. \eqref{hde} admits two algebraically independent solutions $s_1$ and $s_2$, so defined:

\noi a) if $c\notin \mathbb{N}$
\beq s_1:={}_{2} F_{1}(a,b;c;t);
\eeq
\noi b) if $c\notin \mathbb{Z}$
\beq
s_2:=\omega^{1-c}{}_{2} F_{1}(1+a-c,1+b-c;2-c;t),
\eeq
where ${}_{2} F_{1}(a,b;c;x)=\sum_{n=0}^{\infty}\frac{(a)_n (b)_n}{(c)_n}\frac{x^n}{n!}$. For different values of $c\in\mathbb{R}$, the solutions $s_1$ and $s_2$ should be replaced by more complicated expressions. Theorem \ref{main1} provides the following discrete version of eq. \eqref{hde}
\begin{eqnarray}
\nn & & \sum_{k=0}^{n}  \left(\sum_{j=0}^{k-1}(-1)^{k-1-j}\frac{1}{j!(k-1-j)!}- \sum_{j=0}^{k-2}(-1)^{k-2-j}\frac{1}{j!(k-2-j)!}\right)
\\ \nn &\times&\frac{n!}{(n-k)!}\left(z_{j+2}-2z_{j+1}+z_{j}\right) - a b z_{n}+\\ \nn
&+&\left[c-(a+b+1)\sum_{k=0}^{n} \sum_{j=0}^{k-1}(-1)^{k-1-j}\frac{1}{j!(k-1-j)!}\frac{n!}{(n-k)!}\right] \left(z_{j+1}-z_{j}\right)=0. \\ \label{hyperd}
\end{eqnarray}
If we introduce the finite Gauss sum
\beq
 G(n; a,b,c):=\sum_{k=0}^{n}\frac{(a)_k (b)_k}{(c)_k}\frac{n!}{(n-k)!},
\eeq
the two independent solutions of eq. \eqref{hyperd} are provided by
\beq
z^{(1)}_{n}:=G(n; a,b,c), \hspace{10mm} z^{(2)}_{n}:=\omega^{1-c}{}_{2} G(n;1+a-c,1+b-c;2-c).
\eeq

%
\subsection{A nonlinear case} The differential equation
\beq
z'(t)-t^{k} z(t)^2=0
\eeq
possesses the categorical equivalent equation
\beq
z_{n+1}-z_{n}-\sum_{j_1,j_2=0}^{n}\frac{(-1)^{j_1+j_2+n-k}}{(n-j_1-j_2-k)!}\frac{z_{j_1}z_{j_2}}{j_1 ! j_2 !}=0. \label{nonlincase}
\eeq
A two parametric family of solutions of eq. \eqref{nonlincase} is provided by
\beq
z_{n}=\sum_{l=0}^{n}g(m)\frac{n!}{(n-m)!}, \qquad \text{with} \quad g(m)=\frac{1}{m!}\frac{d^{m}}{dt^{m}}\frac{-(k+1)}{t^{k+1}+c_1+k c_2}.
\eeq
\subsection{Difference equations associated with classical orthogonal polynomials}
The previous theory allows to introduce integrable maps admitting as solutions orthogonal polynomials.
Here the examples of this general construction realted to the classical orthogonal polynomials are discussed briefly.
\subsubsection{Discrete Hermite equation}
The differential equations for the Hermite polynomials
\beq
z''(t)-t z'(t)+m z(t)=0, \hspace{10mm} m\in\mathbb{N}
\eeq
by using Theorem \ref{main1} can be discretized in the form
\beq
z_{n+2}-2z_{n+1}+z_{n}-\sum_{k=0}^{n}\sum_{j=0}^{k-1}(-1)^{k-1-j}\frac{1}{j! (k-1-j)!}\frac{n!}{(n-k)!}\left(z_{j+1}-z_{j}\right)+m z_{n}=0. \label{discrHermite}
\eeq
Eq. \eqref{discrHermite} for each value of $m\in\mathbb{N}$ admits as a solution the $m$-th degree polynomial $\frak{h}_{m}(n)$,
obtained by the action of the morphism $\lambda_{\partial, \Delta}$ on the $m$-th degree Hermite polynomials
\beq
H_{m}(x)=\sum_{k=0}^{m}h_{k}x^{k} \longrightarrow \frak{h}_{m}(n):=\sum_{k=0}^{m}h_{k}\frac{n!}{(n-k)!}  .
\eeq
The map \eqref{discrHermite} is a model for the (gauge rotated) Hamiltonian of a \textit{discrete quantum mechanical harmonic oscillator}. In this setting, $m$ plays the role of the quantized energy of the odel, and the polynomials $\frak{h}_{m}$ represent the associated eigenfunctions.


\subsubsection{Discrete Jacobi equation}
By applying the previous technique, the equation
\beq
(1-t^2)z''(t)+\left[\beta-\alpha-(\alpha+\beta+2)t\right]z'(t)+m(m+\alpha+\beta+1)z(t)=0
\eeq
can be discretized in the form
\beqa
  \hspace{-13mm}\nn & &\left(z_{j+2}-2z_{j+1}+z_{j}\right)+(\beta- \alpha)\left(z_{j+1}-z_{j}\right) +m(m+\alpha+\beta+1)z_{n}+ \\
  \nn & &\sum_{k=0}^{n}\sum_{j=0}^{k-2}(-1)^{k-2-j}\frac{1}{j! (k-2-j)!}\frac{n!}{(n-k)!}\left(z_{j+2}-2z_{j+1}+z_{j}\right)+\\
\nn &-& (\alpha+\beta+2)\sum_{k=0}^{n}\sum_{j=0}^{k-1}(-1)^{k-1-j}\frac{1}{j! (k-1-j)!}\frac{n!}{(n-k)!}\left(z_{j+1}-z _{j}\right)=0.
\eeqa
This equation admits as a solution the family of polynomials $\{J_{m}^{(\alpha,\beta)}(n)\}_{m\in\mathbb{N}}$, where
\beq
J_{m}^{(\alpha,\beta)}(n):=\frac{\Gamma(\alpha+m+1)}{m!\Gamma(\alpha+\beta+m+1)}\sum_{k=0}^{m}\binom{m}{k}\frac{\Gamma(\alpha+\beta+m+k+1)}{\Gamma(\alpha+m+1)}
\left(\frac{1}{2}\right)^{k}\frac{(n-1)!}{(n-k-1)!}.
\eeq
Virtually all classes of orthogonal polynomials, including nonclassical and Sobolev-type ones, can be associated with integrable maps in a analogous way. The analysis of other similar cases is left to the reader.
\section*{Acknowledgment}




This research has been supported by the grant FIS2011--22566, Ministerio de Ciencia e Innovaci\'{o}n, Spain.

\end{document}